\documentclass[11pt,runningheads]{llncs}

\usepackage{epsfig}
\usepackage{verbatim}
\usepackage{url}
\usepackage{fullpage}

\def\reverse#1{\hat{#1}}

\def\SA{\mbox{\rm {\sf SA}}}
\def\ISA{\mbox{\rm {\sf ISA}}}

\def\LF{\mbox{\rm {\sf LF}}}

\def\rank{\mbox{\rm {\sf rank}}}
\def\select{\mbox{\rm {\sf select}}}
\def\access{\mbox{\rm {\sf access}}}

\def\X{\mathsf{X}}

\def\B{\mathsf{B}}

\def\A{\mathsf{A}}

\def\Y{\mbox{\rm {\sf Y}}}
\def\Y{\mathsf{Y}}

\def\C{\mbox{\rm {\sf C}}}

\def\W{\mbox{\rm {\sf W}}}
\def\BWT{\mbox{\rm {\sf L}}}

\def\LPF{\mbox{\rm {\sf LPF}}}

\def\+{\!+\!}
\def\-{\!-\!}

\begin{document}
\title{Range Predecessor and Lempel-Ziv Parsing\thanks{This research is
    supported by Academy of Finland through grants 258308 and 284598.}}

\author{
Djamal Belazzougui
\and
Simon J. Puglisi
}

\institute{
    Department of Computer Science,University of Helsinki\\
    Helsinki Institute for Information Technology (HIIT)\\
    Helsinki, Finland\\
    \email{\{belazzou,puglisi\}@cs.helsinki.fi}\\[1ex]
}

\date{}

\maketitle \thispagestyle{empty}
\setcounter{page}{0}

\begin{abstract}
The Lempel-Ziv parsing of a string (LZ77 for short) is one of the most
important and widely-used algorithmic tools in data compression and string
processing.
We show that the Lempel-Ziv parsing of a string of length $n$ on an alphabet of size
$\sigma$ can be computed in $O(n\log\log\sigma)$ time ($O(n)$ time if we allow
randomization) using $O(n\log\sigma)$ bits
of working space; that is, using space proportional to that of the input string in
bits. The previous fastest algorithm using $O(n\log\sigma)$ space takes $O(n(\log\sigma+\log\log n))$ time.
We also consider the important {\em rightmost} variant of the problem,
where the goal is
to associate with each phrase of the parsing its {\em most recent} occurrence in the input
string. We solve this problem in
$O(n(1 + (\log\sigma/\sqrt{\log n}))$ time, using the same working
space as above. The previous best solution for rightmost parsing uses $O(n(1+\log\sigma/\log\log n))$
time and $O(n\log n)$ space.
As a bonus, in our solution for rightmost parsing we provide a faster construction
method for efficient {\em 2D orthogonal range reporting}, which is of independent interest.
\end{abstract}

\newpage

\section{Introduction}
\label{sec-intro}

For almost four decades the LZ parsing~\cite{ZL77} (or LZ77) has been a 
mainstay of data compression, and is widely used today in several 
popular compression tools (such as {\em gzip} and {\em 7zip}), as a part 
of larger storage and search systems~\cite{lucene,fm2010,lemur},
and as a basic measure of information content~\cite{cllppss2005}.

In addition to its long history in compression,
LZ77 also has a wealth of applications in string processing. The factorization
reveals much of the repetitive structure of the input string and this can be exploited
to design efficient algorithms and data structures. For example, optimal algorithms
for computing all the tandem repetitions~\cite{kk1999} and seeds~\cite{kkrrw2012} 
in a string, rely on LZ77.
More recently, LZ77 has become a basis for pattern matching
indexes~\cite{ggknp2014,kn2012}
and compressed data structures~\cite{bggkopt2015}.

Because its computation is a time-space bottleneck in these and other useful and 
interesting applications, efficient algorithms for LZ77 factorization have been 
the focus of intense research almost since its discovery (see~\cite{aciksty2012,kkp2013-sea} for recent surveys).

\smallskip

The LZ parsing breaks a string $S$ up into $z$ phrases (factors). The phrase
starting at position $i$ is either (a) the first occurrence of letter $S[i]$ in 
$S$, or (b) the longest substring starting at position $i$ that has at least one 
occurrence starting to the left of position $i$ in $S$. For example, the LZ 
parsing of string $S = araarraaa$ is $a|r|a|ar|raa|a$. Compression can be achieved by 
replacing phrases of type (b) with a pair of integers $(p_i,\ell_i)$ that indicate respectively the 
starting position and length of a previous occurrence of the phrase in $S$. For example, the 
fifth phrase $raa$ would be represented with $(1,3)$ because substring $S[1,1+3-1] = S[1,3] = raa$
is a prior occurrence.

It is important to note that there is sometimes more than one previous occurrence of a phrase, 
leading to a choice of $p_i$ value. If 
$p_i < i$
is the largest possible for every phrase then we call the parsing {\em rightmost}. 
In their study on the bit-complexity of LZ compression, Ferragina et al.~\cite{fnv2013} 
showed that the rightmost parsing can lead to encodings asymptotically smaller than
what is achievable with other choices of $p_i$.

\paragraph{Main results.}

This article elucidates an important part of the asymptotic landscape of LZ
factorization algorithms. In particular:

\begin{enumerate}
\item We describe an algorithm that computes the LZ factorization in $O(n\log\log\sigma)$
time ($O(n)$ time if randomization is allowed) using only $O(n\log\sigma)$ bits of working space. 
This is the first algorithm for LZ parsing that uses compact space\footnote{Compact space means space within a constant
factor of the size of the input, excluding lower-order terms; so $O(n\log\sigma)$ bits of
space in our case.} and $o(n(\log\sigma+\log\log n))$ time.

\item Our initial approach does not provide any guarantees on the $p_i$
value it computes (other than that it is less than $i$, of course). In Section~\ref{sec-rightmost} 
we consider the problem of ensuring the $p_i$ value computed
is the {\em rightmost} possible (computing the rightmost parsing). For this
problem we describe an algorithm using $O(n\log\sigma)$ bits of space 
and $O(n(\log\log\sigma + \log\sigma/\sqrt{\log n}))$ ($O(n(1 + \log\sigma/\sqrt{\log n}))$ 
if we allow randomization)
time. This significantly improves on the algorithm of Ferragina et al.~\cite{fnv2013},
which runs in $O(n(1 + \log\sigma/\log\log n))$ and anyway requires $O(n\log n)$ bits of space.

\item On the way to our rightmost parsing result we provide a faster
preprocessing method for 2D orthogonal range reporting queries --- 
a classic problem in computational geometry. Given $n$ points 
the data structure requires $O(n\log n)$ bits of space, $O(\log^\epsilon n)$ query time (per reported point), 
and, critically, $O(n\sqrt{\log n})$ time to construct. The data structure 
has the additional property that the reported points are returned in order
of their $x$-coordinate, which is essential to our rightmost parsing algorithm,
and many other applications too (see, e.g.,~\cite{yhw2011,nn2012}). 
Our result is a counterpart to the $O(n\sqrt{\log n})$ 
construction method for 2D orthogonal range {\em counting}, 
by Chan and P{\u a}tra{\c s}cu~\cite{CP2010}.
\end{enumerate}

\paragraph{Roadmap.}
The following section sets notation and lays down basic concepts, data structures,
and results. We then review related work in Section~\ref{sec-related}.
Section~\ref{sec-basic} describes an algorithm for LZ parsing in compact
space that runs in $O(n\log\log\sigma)$ time ($O(n)$ time if 
we allow randomization). 
Section~\ref{sec-rightmost} deals with the rightmost
parsing. Section~\ref{sec-rangepred} --- where we describe our result for range reporting queries --- can be read 
independently of the rest of the document if so desired. 
Open problems are offered in Section~\ref{sec-conclusion}.

\section{Tools}
\label{sec-preliminaries}

Our model in this paper is the word RAM. Space bounds will always be expressed
as the number of bits used by an algorithm. We now define some of the basic tools
we use throughout.

\paragraph{Strings.}
Throughout we consider a string $\X = \X[1..n] = \X[1]\X[2]\ldots
\X[n]$ of $|\X| = n$ symbols drawn from the alphabet $[0..\sigma-1]$.
For technical convenience we assume $\X[n]$ is a special ``end of string'' symbol, \$, smaller than
all other symbols in the alphabet.

The reverse of $\X$ is denoted $\reverse{\X}$.  For $i=1,\ldots,n$ we
write $\X[i..n]$ to denote the {\em suffix} of $\X$ of length $n-i+1$,
that is $\X[i..n] = \X[i]\X[i+1]\ldots \X[n]$.  We will often refer to
suffix $\X[i..n]$ simply as ``suffix $i$''. Similarly, we write
$\X[1..i]$ to denote the {\em prefix} of $\X$ of length $i$.
$\X[i..j]$ is the {\em substring} $\X[i]\X[i+1]\ldots \X[j]$ of $\X$
that starts at position $i$ and ends at position $j$. 
Slightly abusing the notation, a string $\X$ will also refer 
to the integer obtained by considering the bits in decreasing 
significance in left-to-right order ($\X[1]$ is most significant 
bit followed by $\X[2]$ and so on). 

 \paragraph{Suffix Arrays.}
 The suffix array~\cite{mm1993} $\SA_{\X}$ 
 (or just $\SA$ when the context is clear)
 of a string $\X$
 is an array $\SA[1..n]$ which
 contains a permutation of the integers $[1..n]$ such that $\X[\SA[1]..n]
 < \X[\SA[2]..n] < \cdots < \X[\SA[n]..n]$.  In other words, $\SA[j] =
 i$ iff $\X[i..n]$ is the $j^{\mbox{{\scriptsize th}}}$ suffix of $\X$
 in lexicographical order.
 The inverse
 suffix array $\ISA$ is the inverse of $\SA$, that is
 $\ISA[i] = j$ iff $\SA[j] = i$.

 For a string $\Y$, the $\Y$-interval in the suffix array $\SA_{\X}$ is
 the interval $\SA[s..e]$ that contains all suffixes having $\Y$ as a
 prefix. The $\Y$-interval is a representation of the occurrences of
 $\Y$ in $\X$. For a character $c$ and a string $\Y$, the computation
 of $c\Y$-interval from $\Y$-interval is called a \emph{left extension}
 and the computation of $\Y$-interval from ${\Y}c$-interval is called a
 \emph{right contraction}. \emph{Left contraction} and \emph{right
   extension} are defined symmetrically.

 \paragraph{BWT and backward search.}
 The Burrows-Wheeler Transform~\cite{bw1994} denoted $\BWT[1..n]$ is a
 permutation of $\X$ such that $\BWT[i] = \X[\SA[i]-1]$ if $\SA[i]>1$
 and $\$$ otherwise. We also define $\LF[i] = j$ iff $\SA[j] =
 \SA[i]-1$, except when $\SA[i] = 1$, in which case $\LF[i] = \ISA[n]$.
 Clearly, if we know $I = \ISA[1]$, we can {\em invert} the BWT to obtain the original string 
 right-to-left via repeated applications of $\LF$, outputing $\BWT[I]$, then
 $\BWT[\LF[I]]$, then $\BWT[\LF[\LF[I]]]$, and so on. Note that, after $I$, 
 this process also visits the positions in $\SA$ of suffixes $n$, then $n-1$, and so on.
 
 Let $\C[c]$, for symbol $c$, be the number of symbols
 in $\X$ lexicographically smaller than $c$.  The function
 $\rank(\X,c,i)$, for string $\X$, symbol $c$, and integer $i$, returns
 the number of occurrences of $c$ in $\X[1..i]$.  It is well known that
 $\LF[i] = \C[\BWT[i]] + \rank(\BWT,\BWT[i],i)$ and that this ``special''
 form of rank (where $c = \X[i]$) can be computed in $O(1)$ time after
 $O(n)$ time preprocessing to build a data structure of size $n\log\sigma+O(n\log\log\sigma)$~\cite{BN2014}.

 Furthermore, we can
 compute the left extension using $\C$ and $\rank$.  If $\SA[s..e]$ is
 the $\Y$-interval,
 then
 $\SA[\C[c]+\rank(\BWT,c,s),\C[c]+\rank(\BWT,c,e)]$ is
 the $c\Y$-interval.
 This is called \emph{backward search}~\cite{fm2005}. Note that backward
 search either requires general rank queries, which can be answered in $O(\log\log \sigma)$ time after
 $O(n)$ time preprocessing to build a data structure of size $n\log\sigma+o(n\log\sigma)$~\cite{GOR10}
or can be solved using a more sophisticated data structure which occupies the same asymtotic space 
but requires $O(n)$ randomized time preprocessing~\cite{BN2014}.

There are many BWT construction algorithms (see~\cite{pst2007} for a survey),
and some of them operate in compact space. In particular, 
Hon, Sadakane, and Sung~\cite{hss2009} show 
how to construct the BWT in $O(n\log\log\sigma)$ time and $O(n\log\sigma)$ 
bits of working space. More recently, Belazzougui~\cite{b2014} showed how the BWT can be computed 
in $O(n)$ time (randomized) and $O(n\log\sigma)$ bits of working space.

\paragraph{Wavelet Trees.}
Wavelet trees~\cite{GGV03} are a tool from compressed data structures~\cite{nm2007} 
that encode a string $S$ on alphabet in $n\log\sigma + o(n\log\sigma)$ bits and allow 
fast computation of various queries on the original string such as $\access$ to the
$i$th symbol, $\rank$, $\select$, and various range queries~\cite{navarro2014wavelet}. 
The operation $\rank$ is as defined above, while operation $\select(S,c,i)$ 
for symbol $c$, and integer $i$ returns the position of the 
$i$th occurrence of $c$ in $\X$.  

Let $S[1..n]$ be a string of $n$ symbols, where each symbol is in
the range $[1..\sigma]$.  The wavelet tree $\W_S$ of $S$ is a
perfect binary tree with $\sigma$ leaves. The leaves are labelled
left-to-right with the symbols $[1..\sigma]$ in increasing
order. For a given internal node $v$ of the tree, let $s_v$ be the
subsequence of $S$ consisting of only the symbols on the leaves in the
subtree rooted at $v$. We store at $v$ a bitvector $b_v$ of $|s_v|$
bits, setting $b_v[i] = 1$ if symbol $s_v[i]$ appears in the right
tree of $v$, and $b_v[i] = 0$ otherwise. Note that $s_v$ is not
actually stored, only $b_v$. Clearly $\W_S$ requires $n\log\sigma + o(n\log\sigma)$
bits.

\paragraph{LZ77.}
Before defining the LZ77 factorization, we introduce the concept of a
{\em longest previous factor} (LPF).  The LPF at position $i$ in
string $\X$ is a pair $\LPF_{\X}[i]=(p_i,\ell_i)$ such that, $p_i < i$,
$\X[p_i..p_i+\ell_i) = \X[i..i+\ell_i)$, and $\ell_i$ is maximized.
In other words, $\X[i..i+\ell_i)$ is the longest
prefix of $\X[i..n]$ which also occurs at some position $p_i < i$ in
$\X$.

The LZ77 factorization (or LZ77 parsing) of a string $\X$ is then just
a greedy, left-to-right parsing of $\X$ into longest previous
factors. More precisely, if the $j$th LZ factor (or {\em phrase}) in
the parsing is to start at position $i$, then we output $(p_i,\ell_i)$
(to represent the $j$th phrase), and then the $(j+1)$th phrase starts
at position $i+\ell_i$. The exception is the case $\ell_i=0$, which
happens iff $\X[i]$ is the leftmost occurrence of a symbol in $\X$. In
this case we output $(\X[i],0)$ (to represent $\X[i..i]$) and the next
phrase starts at position $i+1$.  When $\ell_i > 0$, the substring
$\X[p_i..p_i+\ell_i)$ is called the {\em source} of phrase
$\X[i..i+\ell_i)$. We denote the number of phrases in the LZ77 parsing
of $\X$ by $z$.

\begin{theorem}[e.g., K{\"a}rkk{\"a}inen~\cite{k1999}]
The number of phrases $z$ in the LZ77 parsing of a string of $n$ symbols
on an alphabet of size $\sigma$ is $O(n/\log_{\sigma}n)$
\end{theorem}

\section{Related Work}
\label{sec-related}

There have been many LZ parsing algorithms published, especially recently. 
Most of these results make no promise about the rank of the previous factor occurrence they 
output.
The current fastest algorithm in practice (ignoring memory constraints) is due to 
K{\"a}rkk{\"a}inen et al.~\cite{kkp2013}. Their algorithm runs in optimal $O(n)$ time and uses $2n\log n + n\log\sigma$ 
bits of space. Goto and Bannai~\cite{gb2014} improve space usage to $n\log n + n\log\sigma$ 
bits, while maintaining linear runtime.

Compact-space algorithms are due to Ohlebusch and Gog~\cite{og2011}, Kreft and Navarro~\cite{kn2012}, 
K{\"a}rkk{\"a}inen et al.~\cite{kkp2013-sea}, and Yamamoto et al.~\cite{yibit2014}. 
All these approaches
take in $O(n\log n)$ time. Very recently a solution with $n\log\sigma+O(n)$ space and $O(n(\log\sigma+\log\log n))$ time has been proposed in~\cite{kosolobov2015faster}. 
We show a significant improvement --- to $O(n\log\log\sigma)$ time --- is possible in the compact
setting. If we allow randomization, then our time becomes linear. 

There are significantly fewer results on the rightmost problem. Early algorithmic work is
due to Amir, Landau and Ukkonen~\cite{alu2002}, who provide an $O(n\log n)$ time and space
solution, but it should be noted that selection of rightmost factors had already been used as a heuristic in the data
compression community for years (for example, in the popular {\em gzip} compressor). Recently,
Larsson~\cite{l2014} showed how to compute the rightmost parsing online in the same
$O(n\log n)$ time and space bounds as Amir et al.. 
Currently the best prior result for rightmost parsing is an
$O(n\log n)$ space and
$O(n + n\log\sigma/\log\log n)$ time
algorithm due to Ferragina, Nitto and Venturini~\cite{fnv2013}. 
We provide a faster algorithm that uses significantly less space.

Our result for rightmost makes use of an improved technique for 
{\em range predecessor queries}\footnote{Elsewhere these queries are variously called 
range successor~\cite{nn2012} and range next value queries~\cite{CIKRTW12}.}.
By combining text indexing with range reporting our thus work continues a long tradition in string 
processing, recently surveyed by Lewenstein~\cite{l2013}.
Given a set of two-dimensional points $P$, 
the answer to an orthogonal range predecessor query $Q = [a, +\infty]\times[c, d]$
is the point $p \in P$ with largest $y$-coordinate among all points that are in the rectangle
$Q$. 

If one is allowed $O(n\log n)$ bits of space, a data structure due to Yu, Hon, and Wang~\cite{yhw2011}
supports range predecessor queries in $O(\log n/ \log \log n)$ time and takes 
$O(n\log n/ \log \log n)$ time to construct. Navarro and Nekrich~\cite{nn2012}
subsequently improved query time to $O(\log^\epsilon n)$,
where $\epsilon$ is an arbitrarily small positive constant, however their structure 
has $O(n\log n)$ construction time~\cite{nn2013-private}.
Multiary wavelet trees are also capable of answering range predecessor queries, and
recently Munro, Nekrich, and Vitter~\cite{MNV14} (and contemporaneously Babenko, Gawrychowski, Kociumaka, and 
Starikovskaya~\cite{BGKS15}) showed how to construct wavelet trees 
in $O(n\log\sigma/\sqrt{\log n})$ time using $O(n\log n)$ bits of working space, and
supporting queries in $O(\log \sigma/\log \log n)$ time.

Finally, we note that a data structure for range predecessor queries immediately implies one 
for the classic and much studied 2D orthogonal range reporting problem from computational 
geometry~\cite{chan2011orthogonal},
in which we seek a data structure to report all points contained in a four-sided query rectangle 
$Q = [a, b]\times[c, d]$. Our range predecessor result is a $O(n \log n)$-bit data structure 
with query time $O(\log^\epsilon n)$ (for any constant $\epsilon<1$)
that can be built in $O(n\sqrt{\log n})$ time. This matches the best 
known query time of for this problem when using $O(n \log n)$ bits of space, 
and to our knowledge is the first to offer construction time $o(n\log n/\log\log n)$.
To put our result in context, Chan and P{\u a}tra{\c s}cu~\cite{CP2010} have shown that a 
2D range {\em counting} data structure with $O(\log n)$ query time and $O(n\log n)$ bits of space can be 
built in $O(n \sqrt{\log n})$ time. 

\section{Lempel-Ziv Parsing in Compact Space}
\label{sec-basic}

Assume the next LZ factor starts at position $i$ in the string. Our basic approach 
to compute the factor is to treat $\X[i,n]$ as a pattern and perform a prefix search
for it in $\X$, which we simulate via backward search steps on the FM-index of the 
reverse text $\X'$. Consider a generic step $j$ of this backward search, in which we 
have a range $[s_j,e_j]$ of the BWT and SA of $\X'$. The factor beginning at $i$ 
has length at least $j$ if and only if $\SA_{\X'}[s_j,e_j]$ contains a value $p_i > i$.
To see this, observe that the presence of such a $p_i$ in $\SA_{\X'}[s_j,e_j]$ means
there is a substring $\X'[p_i..p_i+j]$ that occurs {\em after} substring 
$\X[i..i+j] = \X'[n-i-j..n-i]$ in $\X'$, which in turn implies there is an occurence of 
$\X[i..i+j]$ before position $i$ in $\X$ (starting at position $n-p_i$, in fact).

Our problem now is to be able to determine if $\SA_{\X'}[s_j,e_j]$ contains a value 
larger than $i$ at any given step in the above process. One solution is to preprocess 
SA for range maximum queries. However, this requires that we either first store SA in 
plain form, which requires $O(n\log n)$ bits, or that we obtain the values of SA 
left-to-right $\SA[1], \SA[2], \ldots$ (the order in which they are required for RMQ 
preprocessing) via repeated decoding using the SA sample, requiring $O(n\log n)$ time. 
Either of these straightforward methods uses more space or time than we desire. 

Our approach instead then is to logically divide the BWT into equal-sized blocks 
of size $b = \log n/2$. We then invert the BWT, and during the inversion we record for 
each block the maximum
value of $\SA[i]$ that we see over the whole inversion process. We store an array,
$\A[1,n/b]$ of these block maxima. Storing $\A$ requires $n\log n/b = 2n = O(n)$ bits.
We now build the succinct RMQ data structure of Fischer and Heun~\cite{FischerH11} on the
array $\A$, which requires $3n/b + o(n/b)$ bits (including during construction) and $O(n)$
time. So far we have used $O(n)$ bits and $O(n)$ time for the inversion.

We are now ready to describe the LZ factorization algorithm, which will involve another 
inversion of the BWT. This time we maintain a bit vector $\B[1,n]$. If, during
inversion, we visit position $j$ in the BWT, then we set $\B[j] = 1$. At the same
time (i.e. while) we are inverting we will perform a backward search for the next LZ factor,
using the as yet unprocessed portion of $\X$ as the pattern.

Say we have factorized part of the string and we are now looking for the LZ
factor that starts at $i$. We match symbols of $i$ using backward search. At a given point
in this process we have matched $\ell$ symbols and have an interval of the $\SA$, say $\SA[s,e]$.
We need to decide if there exists a $p < i \in \SA[s,e]$, which will tell us there is an
occurrence of $\X[p..p+\ell] = \X[i..i+\ell]$ before $i$.

$\SA[s,e]$ can be divided into at most three subranges: one that is covered by a series of block
maxima (i.e. a subarray of $A$), and at most two small subranges at each end, $[s,s']$
and $[e',e]$, each of size at most $\lfloor \log n/2\rfloor$. We compute the maximum value
covered by the block maxima in $O(1)$ time using the RMQ data structure we built over $A$.
For the two small subranges at each end of $\SA[s,e]$ that are not fully covered by block
maxima we consult $B$. If there is a bit set in either of the $\B[s,s']$ or $\B[e',e]$ then
we know there is some suffix in there that is greater than $i$ (because it has already been
visited in the inversion process). Because the (sub)bitvectors 
$\B[s,s']$ or $\B[e',e]$ are so small ($< \lfloor \log n/2 \rfloor$), we can use
a lookup table to determine if there is a set bit in $O(1)$ time, and further we can have
the lookup table return the position of one of these set bits. In this way we are able to 
determine in constant time whether we have reached the end of the current factor.

Having determined the length of the current factor, it is then simply a matter 
of using a sampled SA of size $O(n/log n)$ elements that allows us to extract 
arbitrary SA elements in $O(\log n)$ time~\cite{fm2005} to obtain one of the candidate SA 
values from the previous round 
in the search for the current factor (so that we can output a $p_i$ value for the factor). 
This takes $O(\log n)$ time per factor, which over all the $z = O(n/\log_{\sigma} n)$
factors takes $O(n\log \sigma)$. However, runtime can be further reduced to $O(n)$ over all 
factors if we first record the positions of the candidate $p_i$ values for each factor 
(using $O(n\log \sigma)$ bits of space) and obtain them all in a single further inversion 
of the BWT.

As described, our factorization algorithm requires $O(n)$ time and $O(n\log \sigma)$ bits 
of space in addition to the resources needed to construct the BWT and perform $n$ 
backward search steps. We thus have the following theorem.

\begin{theorem}
\label{lz-compact-space}
Given a string $S$ of $n$ symbols on an ordered alphabet of size $\sigma$ we can compute 
the LZ factorization of $S$ using $O(n\log \sigma)$ bits of space and $O(n\log\log \sigma)$ 
time or $O(n)$ time (randomized).
\end{theorem}

\section{Faster preprocessing for range-predecessor queries}
\label{sec-rangepred}

In the {\em range predecessor} problem in rank space, we are to preprocess 
$n$ points on a $[1,n]\times [1,n]$ grid, where all points  
differ in both coordinates, so as to answer to the following 
kind of query: given integers $x_1,x_2$ and $y$ find the point $(u,v)$
such that $x_1\leq u\leq x_2$, $v\leq y$ and $v$ is maximal. 

Navarro and Nekrich presented a solution to this problem~\cite{nn2012} 
that uses space $O(n\log n/\epsilon)$ bits and  
answers queries in time $O(\log^\epsilon n)$.
However, they did not show how to efficiently construct their data structure. 


We now show how 
to efficiently build a variant of the known solutions 
for range predecessor. The solution we describe here has query 
time $O(\sqrt{\log n}\log\log n)$. We later show how to generalize it 
to have query time $O(\log^\epsilon n)$ for arbitrary $0<\epsilon<1$. 

We start by defining the sequence $Y$ of length $n$ over alphabet $[1..n]$ 
obtained by setting $Y[x]=y$ for every point $(x,y)$ in the set of input points. 
We similarly define the sequence $X$ such that $X[y]=x$ for every input 
point $(x,y)$.

At a high-level, the solution uses a top-level data structure that resembles 
a multiary wavelet tree~\cite{FMMN07} with arity $2^{\sqrt{\log n}}$
and depth $\sqrt{\log n}$. We note that a standard wavelet tree can answer 
range predecessor queries in time $O(\log n)$. 
Without loss of generality, we assume that $\sqrt{\log n}$ is 
integral and that $\log n$ is divisible by $\sqrt{\log n}$. 
The top-level data structure is a tree 
with $\sqrt{\log n}$ levels. The arity of the tree is exactly 
$2^{\sqrt{\log n}}$. 
At any level $i\in[0..\sqrt{\log n}-1]$, we will have $2^{i\sqrt{\log n}}$
nodes labelled with values $[0..2^{i\sqrt{\log n}}-1]$
(note that at level $0$ we will only have the root 
node, which is labelled by $0$). 
Any node $\alpha$ at level $i$ will have as children all nodes 
$\beta$ at level $i+1$ such that $\beta[1..i\sqrt{\log n}]=\alpha$.

To every node $\alpha$ in the tree we associate 
a sequence $Y_\alpha$ of length $n_\alpha$. 
The sequence $Y_\alpha$ will be over alphabet $[1..2^{\sqrt{\log n}}]$,  
while the sequence $Y'_\alpha$ is a sequence of integers from $[1..2^{(\sqrt{\log n}-i)\sqrt{\log n}}]$. 
Every node of the tree will contain the following substructures: 
\begin{enumerate}
\item A plain representation of the sequence $Y_\alpha$. 
This sequence occupies $n\sqrt{\log n}$ bits of space. 
\item A regular wavelet tree~\cite{GGV03} $W_{\alpha}$ over the sequence $Y_\alpha$. This wavelet tree will have 
depth $\sqrt{\log n}$. It can be used to answer to range predecessor queries over the sequence $Y_\alpha$
in time $O(\sqrt{\log n})$. 
\item Exactly $2^{\sqrt{\log n}}$ predecessor data structures that support 
predecessor ($\rank$) queries in $O(\log\log n)$ time. For each character, $c\in[1..2^{\sqrt{\log n}}]$, 
we store its positions of occurrence in $Y_\alpha$ in a predecessor data structure
denoted $P_{(\alpha,c)}$. 
The predecessor data structures are implemented using Elias-Fano 
data structure in such a way that they occupy in total $O(n_\alpha\sqrt{\log n})$ bits of space. 
\item A range minimum query data structure on the sequence $Y_\alpha$ denoted $\mathtt{Rmin}_\alpha$. 
The data structure will occupy $O(n_\alpha)$ bits and answer queries in constant tile. 
\item A range maximum query data structure on the sequence $Y_\alpha$ denoted $\mathtt{Rmax}_\alpha$.
This data structure will also occupy $O(n_\alpha)$ bits and answer queries in constant tile. 
\end{enumerate}
All data structures of a node $\alpha$ will use $O(n_\alpha\sqrt{\log n})$ bits 
of space  (for predecessor data structures we count the space as if it was 
one single structure)
except for range minimum and range maximum 
data structures, which will use $O(n_\alpha)$ bits. 
A detailed description of the predecessor data structure 
is given in Appendix~\ref{sec:Elias_Fano_pred}. The space for all nodes 
at the same level will sum up to $O(n\sqrt{\log n})$ and since we have 
$\sqrt{\log n}$ levels, the total space will sum up to $O(n\log n)$
bits. 

We now define how the sequences $Y'_\alpha$ and $Y_\alpha$ 
are built\footnote{Note that the sequence $Y'_\alpha$ is not used 
explicitly in the data structure. It will however be used later, 
when we show how the data structure is queried and constructed.}. 
At any level $i\in[0..\sqrt{\log n}-1]$ we will have $2^{i\sqrt{\log n}}$
nodes. For $\alpha\in[0..2^{i\sqrt{\log n}}-1]$, the sequence $Y'_\alpha$
is built by first constructing the subsequence $Y''_\alpha$ of $n_\alpha$ values in $Y$ whose $i\sqrt{\log n}$
most significant bits equal $\alpha$ (i.e. $Y[1..i\sqrt{\log n}]=\alpha$), 
and then removing the most significant $i\sqrt{\log n}$ bits from every element in $Y''$
(that is $Y'[j]=Y''[j][i\sqrt{\log n}+1..\log n]$ for all 
$j\in[1..n_\alpha]$). 
Then $Y_\alpha$ is obtained from 
$Y'_\alpha$ by taking the most significant $\sqrt{\log n}$ bits from every element 
of $Y'_\alpha$ (that is $Y[j]=Y'[j][1..\sqrt{\log n}]$ for all 
$j\in[1..n_\alpha]$). Notice that for the root node we will have $Y'_0=Y$. 
The total number of nodes will be dominated by the nodes at the lowest level, 
summing up to $\Theta(n/2^{\sqrt{\log n}})$. 

With our data structure now defined, we will next show how to construct it efficiently. 
The description of how queries are answered is given in Appendix~\ref{sec:range_succ_queries}. 

\subsection{Construction of the range-predecessor data structure}

Building all subsequences $Y_\alpha$ and wavelet trees $W_\alpha$ can be done in time $O(n\sqrt{\log n})$
using a variation of the algorithm shown in~\cite{MNV14,BGKS15}. Details are shown in 
Appendix~\ref{sec:wavelet_tree}. The construction of Elias-Fano data structures can also 
easily be done in overall time $O(n\sqrt{\log n})$. Details are shown in Appendix~\ref{sec:simple_Elias_Fano_build}.

Each range minimum and maximum query data structure 
can be constructed in time $O(n_\alpha)$ using the algorithm of Fischer and Heun~\cite{FischerH11}.
When summed over all the nodes $\alpha$, the construction  
and range minimum (maximum) data structures takes time $O(n\sqrt{\log n})$, since the total 
number of elements stored in all the structures is $O(n\sqrt{\log n})$. 
We have thus proved the following theorem.

\begin{theorem}
\label{range_pred_theo1}
Given $n$ points from the grid $[1,n]^2$, we can in 
$O(n\sqrt{\log n})$ time build a data structure that occupies 
$O(n\log n)$ bits of space and that answers 
range predecessor queries in $O(\sqrt{\log n}\log\log n)$ time.
The construction uses $O(n\log n)$ bits of working space. 
\end{theorem}

We can generalize the data structure as follows. 

\begin{theorem}
\label{range_pred_theo2}
Assume that we have available space $N$ 
and preprocessing time $O(N)$ with word-length 
$w\geq \log N$. Then given $n$ points from the grid $[1,n]^2$, 
with $n<N$ and a parameter $c\geq 2$, we can in $O(n(\log n/\sqrt{\log N}+c))$ time build 
a data structure that occupies $O(cn\log n)$ bits of space and that answers range predecessor 
queries in $O(c\log^{(1/c)} n\log\log n)$ time. 
The construction of the data structure uses $O(n\log n)$ bits of working space 
and a precomputed global table 
of size $o(N)$ that can be built in $o(N)$ time. The precomputed table can be shared 
by many instances of the data structure. 
\end{theorem}
\begin{proof}
The proof is more involved. To prove the result we will use multiple levels 
of granularity. 
At the top level, we will have a tree of $\log^{1/c}n$ (tree) levels
(to avoid confusion we call these {\em tree levels}), where 
each node handles a subsequence of $Y$ over alphabet $[1..2^{\log^{(c-1)/c}n}]$
~\footnote{We assume without loss of generality that 
$\log^{1/c}n$ is integral}. 
The level of granularity of this tree is $\log^{(c-1)/c}n$. 
For each node the data structures are exactly the same as the ones in Theorem~\ref{range_pred_theo2}, 
except that the wavelet tree of each sequence is replaced by tree at level of granularity 
$\log^{(c-2)/c}n$, which contains $\log^{1/c}n$ (tree) levels, 
each of which handles a sequence over alphabet $[1..2^{\log^{(c-2)/c}n}]$. 
The recursion continues in the same way until we get to trees at level of granularity~1, 
which is implemented using a wavelet tree. 
Queries are answered in two phases. In the first we determine the longest common 
prefix between the $y$ coordinate of the query and the $y$ coordinate of the answer 
by traversing trees at decreasing levels of granularities
(tree at level $\log^{(c-1)/c}n$, then level of granularity $\log^{(c-2)/c}n$ and so on)
and querying the range-maximum and predecessor data structures at each traversed node. 
Then the remaining bits of the answer are determined by traversing trees of increased levels 
of granularity, querying range-maximum and predecessor data structures. 
The time bound $O(c\log^{1/c}n\log\log n)$, follows because we have 
$c$ levels of granularity and at most $O(\log^{1/c} n)$
nodes are traversed at each level of granularity, where queries at each node 
cost $O(\log\log n)$ time.
Details of how queries are answered are given in Appendix~\ref{subsec:faster_range_pred}.

The main challenge is to quickly construct the Elias-Fano data structures. This 
is shown in Appendix~\ref{sec:bit_parallel_Elias_Fano_construction}. 
The construction for range minimum (maximum) 
or queries is shown in Appendix~\ref{sec:sampled_rmq}. 
Both construction methods make use of bit-level parallelism 
to accelerate processing. 
\qed
\end{proof}
As an immediate corollary, we have the following: 

\begin{corollary}
\label{range_pred_corollary}
Given $n$ points from the grid $[1,n]^2$, for any integer $c\geq 2$, 
we can in $O(n\sqrt{\log n})$ time build a data structure that occupies 
$O(cn\log n)$ bits of space and that answers 
range predecessor 
queries in $O(c\log^{(1/c)} n\log\log n)$ time. 

\end{corollary}

\section{Rightmost Parsing}
\label{sec-rightmost}

We will now apply the construction for range predecessor detailed above 
to obtain a faster algorithm for computing rightmost previous occurrences 
of LZ phrases. An algorithm by Ferragina et al.~\cite{fnv2013} achieves 
$O(n(1 + \log\sigma/\log\log n))$ time, but requires $O(n\log n)$ bits 
of space that can not trivially be reduced.

In this section, we will achieve time $O(n(1 + \log\sigma/\sqrt{\log n}))$
and $O(n\log\sigma)$ bits of space, significantly improving the bounds 
achieved by Ferragina et al.~\cite{fnv2013}. We first present a preliminary 
solution from~\cite{fnv2013} in Section~\ref{sec:basic_rightmost} and then present an initial version of our 
solution that uses $O(n\log n)$ space (sections~\ref{sec:rightmost_long_factors},~\ref{sec:rightmost_sparse_tree}, and~\ref{sec:full_picture_and_opt_space}). 
The solution works by decomposing the phrases into $3$ categories. 
Finding the rightmost occurrences will be fast 
for different reasons. For the first two categories, the reason 
is that the number of phrases is small, while for the last the reason 
is that the rightmost occurrence for each individual phrase will be easier 
to find. 

We divide the full range $[1..n]$ in the suffix array into blocks of a certain size $B$. 
Phrases longer than 
a certain length $\ell$  are handled in Section~\ref{sec:rightmost_long_factors}), phrases 
shorter than a certain length $\ell$ whose suffix array ranges cross a block boundary
are handled in Section~\ref{sec:rightmost_sparse_tree}. Finally the remaining phrases are handled 
in Section~\ref{sec:rightmost_rem_factors}. 

\subsection{Basic solution of Ferragina et al.}
\label{sec:basic_rightmost}

We present the basic solution of~\cite{fnv2013}. 
The original algorithm uses $O(n\sigma)$ time and 
$O(n\log n)$ space. Here, we describe a slightly improved 
version that uses only $O(n\log\sigma)$ bits. 
The algorithm works as follows. To each phrase, we can associate a leaf 
and an internal node in the suffix tree. 
The leaf corresponds to the suffix (text position) that starts at the same 
position as the phrase and the internal node corresponds to the 
range of suffixes prefixed by the phrase. We mark in the suffix 
tree all the internal nodes that are associated with phrases. 
To each leaf, we associate the nearest marked ancestor. This requires $O(n)$ bits 
and $O(n)$ time in total. Also, to each leaf associated with a phrase, we keep a pointer 
to the internal node that corresponds to that phrase. This association can be 
coded in $O(n\log\sigma)$ bits, since we can use a bitvector to 
mark leaves together with an array of at most $z = O(n/\log_\sigma n)$
pointers, each of $O(\log n)$ bits. All the structures can be built 
in $O(n)$ time. 
We keep only the marked nodes in the suffix tree. To each 
marked node, we keep all the phrases that correspond to it. 
To each marked node $\alpha$, we keep a text position $p_\alpha$
that will point to the rightmost position among all the leaves that 
have node $\alpha$ as their nearest marked ancestor. 
Overall the space occupied by all data structures is $O(n\log\sigma)$
bits. The algorithm works by scanning the inverse suffix array
in left-to-right order. That is, at every step $i$, we extract the 
suffix array position (suffix tree leaf) that points to text position $i$,
and update the variable $p_\alpha$, where $\alpha$ is the nearest marked ancestor 
of the leaf. When we arrive at a leaf that corresponds to a phrase,  
we go to the corresponding node, and then explore the entire subtree under that 
node and take the maximum of all variables $p_\alpha$ for all nodes $\alpha$
in the subtree. The scanning of the inverse suffix array, can be done 
by inverting the Burrows-Wheeler transform in increasing text order. 
This can be done using the select operation which can be answered 
in constant time. Thus the inversion takes time $O(n)$. 
The overall space is $O(n\sigma)$. 
The time bound comes from the 
fact that each internal node corresponds to a phrase of length $m$, and can only 
have at most $m$ ancestors. Since at most $\sigma$ phrases are associated with each 
of the ancestors, the node can only be explored $\sigma m$ times: $\sigma$ times 
for each of its $m$ ancestors. Since the total length of the phrases is $n$, 
we conclude that the total running time is $O(n\sigma)$. 
We thus have the following theorem.

\begin{theorem}[space improved from~\cite{fnv2013}]
We can find the right-most positions of all phrases in time $O(n\sigma)$
and working space $O(n\log\sigma)$ bits. 
\end{theorem}

\subsection{Long factors}
\label{sec:rightmost_long_factors}

Because factors do not overlap, there can only be $O(n/\ell)$ factors of length at least $\ell$.
We thus can afford to use time $O(\ell)$ to find the rightmost occurrence 
of each factor. We sample every $r$th position in the text
($r$ and $\ell$ will be set later). We then build a five-sided $3D$ range maximal query data structure
as follows. We will have the text $\X[1..n]$ with split points at positions 
$ir,(i+1)r\ldots $. We then store $n/r$ points as follows. 
For every $i\in[1..n/r]$, we store point $(x,y,i)$ where 
$x$ represents the lexicographic rank of the reverse of substring $\X[(i-1)r+1..ir]$
among all substrings $\X[(i-1)r+1,ir]$, and $y$ the rank of the suffix $\X[ir+1..n]$. A query will consist 
of a triplet $([x_1,x_2],[y_1,y_2],z)$ and will return the point
$(x,y,i)$ with maximal coordinate $i$ among all points that satisfy 
that $x\in[x_1,x_2]$, $y\in[y_1,y_2]$ and $z<i$. 
In this way we store $n'=O(n/r)$ points in total. 
We store the set $S$ of reverse of 
all substrings $\X[(i-1)r+1,ir]$ for $i\in[1..n/r]$ 
in a table $T_S$ sorted in lexicographic order. 
Given any string $p$, we can to determine the range of elements 
of $S$ which have reverse of $p$ as a prefix. The table $T_S$ 
can be built in $O((n/r)(\log n+\frac{r\log\sigma}{\log n}))=O(n(\frac{\log n}{r}+\frac{\log\sigma}{\log n}))$. 
Reverting every string of $p$ can be done by using a lookup table $\mathtt{LT}$ which stores
the reverse of every possible string of length $\log_\sigma n/2$. The space used 
by the lookup table will be $O(\sqrt{n}\log n)$ bits and will allow to revert 
every string of length $\log_\sigma n/2$ in constant time.

The data structure we use occupies $O(n'\log (n'))=O(n\log^2 n/r)$
bits of space and answers queries in time $O(\log^2 n'\log \log n')$. 
This is obtained by building $\log n$ data structures 
for $2D$ range maximal queries~\cite{farzan2012succinct}. 
By building a perfect binary search tree on the third dimension 
$z$, then building a 2D range maximum query data structure for all points 
that fall in one subtree (we use only coordinates $x$ and $y$), one 
multiplies the space and the query time by factor $\log n'$. 
Since the original data structure uses space $O(n')$ words 
and answers in $O(\log n'\log\log n')$ time, 
we obtain the bounds above 
by multiplying the time and space bounds by $\log n'$. 
By replacing $n'$ by $n/r$, the total space usage is
$O(n\log^2n/r)$ bits and the query time is $O(\log^2n\log\log n)$. 

Given a factor $p$ of length at least $\ell$, we will 
issue $r$ queries each of which will take 
$O(\log^2n\log\log n)$ time. The final result will be maximum 
over all the results of the queries. 
In order to determine the queries 
to the $3D$ range-max structure, we will binary search 
the table $T_S$ for every suffix of $p$  of length $i\in[1..r]$
(we first revert the suffix in time $O(r\frac{\log\sigma}{n})$ 
using the table $\mathtt{LT}$). This will determine 
the range $[x_1,x_2]$. The ranges $[y_1,y_2]$ are determined by querying 
the BWT of $\X$ in total time $O(\ell)$ (by backward searching). 

Thus the total query 
time will be $O(r(\log^2n\log\log n+r\frac{\log\sigma}{n})+\ell)$ and the space $O(n\log^2n/r)$. 
Choosing $\ell\geq \log^5 n$ and $r=\log^2 n$ ensures that the total 
time per factor is $O(\log^4\log\log n+\ell)$ which amortizes to $O((\log^4\log\log n+\ell)/\log^5 n+1)=O(1)$
time per character of the factor. 
The total space is dominated by the space used by $T_S$ 
which is $O(n\log\sigma)$ bits, and the total preprocessing time is dominated by the time needed 
to construct $T_2$ which is  $O(n(\frac{\log n}{r}+\frac{\log\sigma}{\log n}))=O(n\frac{\log\sigma}{\log n})$.

\subsection{Sparsified tree}
\label{sec:rightmost_sparse_tree}
If we divide the universe $x[1..n]$ into blocks 
of equal size $B$ and moreover only solve queries
for factors  
whose suffix array range crosses a boundary and whose phrase 
lengths is at most $\ell$, then the number of nodes considered 
can not be more than $O((n\frac{\ell}{B}))$. To justify 
this, consider for every boundary the deepest node that crosses 
a specific boundary. Obviously this node is unique, since if two nodes 
cross the same boundary, one has to be parent of the other
and then one of them would not be the deepest. 
Thus there can be not more than $O(n/B)$ such nodes. 
We call those nodes {\em basic nodes}. On the other hand, 
any node that crosses a boundary has to be ancestor of one of the basic 
nodes. Since, by definition a basic node can not have more than $\ell$ ancestors, 
we deduce that the total number of nodes is $n'=O((n\frac{\ell}{B}))$. 
Recall now that the algorithm described in Section~\ref{sec:rightmost_sparse_tree} traverses 
the tree of phrases and for each leaf updates the minimum of the nearest 
marked ancestor and then for each phrase computes the rightmost pointer 
by traversing the whole subtree under the node of that phrase. 
Since, there are at most $\sigma$ phrases per node, 
each of the $n'$ nodes will be traversed $O(\ell\sigma)$ 
times, at most $\sigma$ times for each of its (at most) $\ell$ ancestors. 
Thus, the total cost will be $O(n+n'\ell\sigma)=O(n+n\frac{\ell^2\sigma}{B})$. 
Choosing $B=\ell^2\sigma$ ensures $O(n)$ overall running time. 
The total additional used space will be $O(n)$ bits dominated by the space
needed to store the nearest-marked ancestor information (see Section~\ref{sec:basic_rightmost}). 

\subsection{Remaining factors}
\label{sec:rightmost_rem_factors}

We will use Theorem~\ref{range_pred_theo2} for short factors 
that do not cross a block boundary. 
For each block we build a range-predecessor data structure. 
We can use parameter $N=n$ and use a global precomputed table that adds 
$o(n)$ bits of space. Since each block contains at most
$B$ points, construction takes 
$O(B\lceil\log B/\sqrt{\log n}\rceil)$ time. 
Each query is solved in time $O((\log B)^{1/c}\log\log n)$. 
Choosing $B=\ell^2\sigma$ means total construction time 
adds up to $O(n\lceil \log \sigma/\sqrt{\log n}\rceil)$.
This dominates the total query time, which adds up to 
$O(\frac{n}{\log_\sigma n}\log^{1/c}(\ell^2\sigma)\log\log n)$. 
Notice that the $y$ coordinates in each block are originally in $[1..n]$. 
In addition to the range-predecessor structure, we will use a predecessor 
structure to reduce the $y$ coordinate of a query to the interval $[1..B]$. 
For that we assume that we have available all the
values $y$ coordinates of the points that fall in the block sorted 
in increasing order. We also assume that the $y$ coordinates 
of the points stored in the range-predecessor data structure have been reduced  
to the interval $[1..B]$. That is, instead of storing the original $y$ coordinate
of each, we store the rank of that coordinate among all values of $y$ coordinates 
that appear in the block. 

\subsection{Putting pieces together and getting optimal space}
\label{sec:full_picture_and_opt_space}
Combining together the three categories above, we can get 
total time $O(n(1+\frac{\log\sigma}{\sqrt{\log n}}))$ and space $O(n\log n)$ bits. Details 
are shown in Appendix~\ref{sec:time_efficient_rightmost}. 
The space can be reduced to optimal $O(n\log\sigma)$ bits. 
This is shown in Appendix~\ref{sec:opt_space_rightmost}. 
We thus have proved the following.
\begin{theorem}
We can find the rightmost occurrences of Lempel-Ziv factors in 
time $O(n(\log\log\sigma+\log\sigma/\sqrt{\log n}))$ and space $O(n\log\sigma)$ 
bits. The time is $O(n(1+\log\sigma/\sqrt{\log n}))$ if randomization 
is allowed. 
\end{theorem}

\section{Conclusions and Open Problems}
\label{sec-conclusion}

We leave two main open problems. Firstly, is it possible to compute the rightmost parsing
in $O(n)$ time, independent of the alphabet size? Note that even using $O(n\log n)$ memory 
and $O(n)$ time would be interesting. The algorithm introduced in Section~\ref{sec-rightmost} 
is the current fastest running in $O(n(1 + \log\sigma / \sqrt{\log n}))$ time (and using 
compact space). Secondly, are the time bounds we achieve, or anything $o(n\log n)$ for 
that matter, possible if processing must be online?

\section*{Acknowledgements} 
Our deep thanks go to: Hideo Bannai for feedback 
on an early rendering of the algorithm of Section~\ref{sec-basic};
and to Juha K{\"a}rkk{\"a}inen, Dominik Kempa, and Travis Gagie, for frequent,
wide-ranging and inspiring discussions on LZ parsing and its properties.

\newpage

\bibliographystyle{splncs03}
\bibliography{lz}
\newpage
\appendix

\section{Wavelet tree construction}
\label{sec:wavelet_tree}
We first describe the core procedure used to build the wavelet tree and 
then describe the wavelet tree construction itself. We then describe
how range-predecessor queries are solved using the wavelet tree. 

\subsection{Core procedure}
\label{subsubsec:core_procedure}
Suppose that we have are given a parameter $N\leq 2^w$ bits, and that we 
can spend $o(N)$ time preprocessing to build (universal) tables that occupy $o(N)$ bits of space. 
The core procedure to build the wavelet tree~\cite{CP2010,MNV14,BGKS15} is as follows.
We are given an array of integers $V[1..m]$ and a bit position $p$, where 
$V[i]\in[0,\sigma]$ for all $i\in[1..m]$, 
and the goal is to produce two arrays $V_0[1..m_0]$ and $V_1[1..m_1]$ 
such that $V_0$ is the subsequence of integers of $V$ whose $p$th 
most significant bit equals $0$ and $V_1$ is the subsequence of elements of $V$ whose 
$p$th most significant bit equals $1$. We will describe a procedure that runs 
in time $O(\lceil\frac{m\log\sigma}{\log N}\rceil)$.

In order to fill the vectors $V_0$
and $V_1$ we use two counters $C_0$ and $C_1$,
initially set to $1$. 
We scan the array $V$ and read it in 
blocks of $B=\frac{\log N}{2\log\sigma}$ elements. Suppose that 
a block contains $t_0$ elements whose bit number $p$ equals~$0$
and $t_1=B-t_0$ whose bit number $p$ equals~$1$. 
We denote by $b$ the block and by $b_0$ (resp. $b_1$), 
the subsequence of elements of $b$ whose bit number $p$ equals 
$0$ (resp. $1$). We append the blocks $b_0$ (resp. $b_1$) 
at the end of $V_0$ (resp. $V_1$) respectively at positions 
indicated by $C_0$ (resp. $C_1$) and then set $C_0=C_0+t_0$
(resp. $C_1=C_1+t_1$). 

We will use a lookup table $L[1..\sigma^B]$ that produces tuples 
$(b_0,t_0,b_1,t_1)$ for every possible 
block $b$ of $B$ elements of length $\log\sigma$ bits each. 
Since we have 
$\sigma^B=O(\sqrt{N})$ possible blocks and every 
element in the table uses $O(\log N)$ bits, the size of the lookup table will 
be $O(\sqrt{N}\log N)$ bits. 

Reading a block $b$, reading the entry $L[b]$, incrementing $C_0$ ($C_1$), and appending $b_0$ ($b_1$) 
to $V_0$ ($V_1$), can all be done in constant time, since each of
the blocks $b,b_0,b_1$ fit in $\log N/2<w$ bits and reading or writing 
blocks of this size requires only a constant number of bit shifts and bitwise logical operations. 

\subsection{Construction}
A wavelet tree for a sequence $S[1..m]$ over alphabet $\sigma$ can be built 
in time $O(m\frac{(\log\sigma)^2}{\log N}+\sigma)$ by repeating the 
core procedure for each node of the wavelet tree. 

More precisely at first level, 
we use the procedure to produce two arrays $S_0[1..m_0]$ and $S_1[1..m_1]$ 
such that $S_0$ is the subsequence of integers of $S$ whose 
most significant bit equals $0$ and $S_1$ is the subsequence of elements of $S$ whose 
most significant bit equals $1$. Notice that the bitvector stored at the root 
can be trivially obtained from $S$ in $O(m)$ time, by just scanning $S$ and 
extracting the most significant 
bit of every element of $V$ and append it to the bitvector. This process can be accelerated 
to run in $O(m\frac{\log\sigma}{\log N})$, by using again a lookup table that givens the 
sequence of $B$ most significant bits for every possible blocks of $B=\frac{\log\sigma}{2\log N}$
blocks of characters. 

At the second level, we will apply the same algorithm to $S_0$
($S_1$), to get the subsequence to produce two arrays $S_{00}[1..m_0]$ and $S_{01}[1..m_1]$
($S_{10}[1..m_0]$ and $S_{11}[1..m_1]$), 
such that $S_{00}$($S_{10}$) is the subsequence of integers of $S$ whose 
second most significant bit equals $0$ and $S_{01}$($S_{11}$) is the subsequence of elements of $S$ whose 
second most significant bit equals $1$. At that point, we can generate the bitvector $b_0$ ($b_2$) 
that contains the second 
most significant bit of $S_0$ ($S_1$) in time $O(|S_0|\frac{\log\sigma}{\log N}+1)$ ($O(|S_1|\frac{\log\sigma}{\log N}+1)$) and finally throw $S_0$ ($S_1$). The generation of those two bitvectors (which are to be stored at the two children of the root of the tree) can also be done in total time $O(m\frac{\log\sigma}{\log N}+1)$. Once a bitvector has been generated 
we index it so as to support $\rank$ and $\select$ queries. This is done in times $O(|S_0|\frac{\log\sigma}{\log N}+1)$ and $O(|S_1|\frac{\log\sigma}{\log N}+1)$ respectively for $S_0$ and $S_1$ using the technique described in~\cite{BGKS15}, which uses lookup tables of size $o(N)$ bits. 
We continue applying the algorithm in the same way for every node at every level, until we get to a leaf of the wavelet tree. 

Since we have $\sigma$ nodes and $\log\sigma$ levels, the total running time is 
$O(m\frac{(\log\sigma)^2}{\log N}+\sigma)$. The total space used for the lookup tables is 
$O(\sqrt{N}\log N\log\sigma)$ bits and the total temporary space used during the construction 
is $O(m\log\sigma)$ bits, since only bitvectors are kept after a given level is constructed. 

We thus get the following lemma: 

\begin{lemma}
\label{lemma:wavelet_tree_build}
Given a sequence $Y$ of length $m$ over alphabet $[1..\sigma]$ 
and global precomputed tables of total size $o(N)$, where $\sigma\leq m\leq N\leq 2^w$, we can build the wavelet tree
over $Y$ in time $O(m\frac{(\log\sigma)^2}{\log N}+\sigma)$, 
using $O(n\log\sigma)$ bits of temporary space. 
\end{lemma}

\subsection{Range-predecessor queries using wavelet tree} 

We now show how range-predecessor queries are solved using a wavelet tree. 
We are given integers $x_1,x_2$ and $y_2$ and must find the point $(x,y)$
such that $x_1\leq x\leq x_2$, $y\leq y_2$ and $v$ is maximal. 
We assume that there is no point $(x,y_2)$ such that $x\in[x_1,x_2]$. 
Otherwise, the answer to the query is trivial. 
The query proceeds in two phases. In the first phase, we find the longest 
common prefix between $y_2$ and $y$, and in second phase, we determine the 
remaining bits of $y$. At this second phase, a bit number $i$ is determined 
by traversing a node at level $i$ and its value is the maximal value
for which the query issued at that node gives a non-empty answer 
(queries and their answers are defined more precisely below). 

The first phase proceeds as follows. At the root level, we 
check whether there interval $[x_1,x_2]$ in the bitvector $b$ stored at the root 
contains an occurrence of $y[1]$, 
by checking that $\rank(b,y[1],x_1-1)<\rank(b,y[1],x_2)$. 
If that is the case, we continue to child number $y[1]$ of the root 
and recurse using the interval $[x'_1,x'_2]=[\rank(b,y[1],x_1-1)+1,\rank(b,y[1],x_2)]$. 
Let $b_{y[1]}$ be the bitvector associated with child number $y[1]$. 
At the next level, we use two $\rank$ queries on $b_{y[1]}$ with symbol $y[2]$ and
points $x'_1$ and $x'_2$ and check whether 
the interval $[\rank(b_{y[1]},y[2],x'_1-1)+1,\rank(b_{y[1]},y[2],x'_2)]$ is non-empty. 
We continue in the same way down the tree, until we reach a node $\alpha$ for which we have an empty 
interval. Let $i$ be the level of that node. Suppose that $y_2[i]=1$, then the answer is in the subtree 
of $\alpha$, and we can deduce that the longest common prefix between $y_2$ and $y$
of $i-1$. If $y_2[i]=0$, then we go up the tree decrementing $i$ at each step and at a node $\alpha$ at level $i$, check whether $y_2[i]=1$, and if so requery the bitvector at that node with the interval with 
which we already queried it before, but this time with bit 
value $0$. If the query is successful, we stop climbing the tree and we will have determined that the longest 
common prefix between $y$ and $y_2$ is $i-1$. 

We now describe the second phase of the query. 
The remaining $\log\sigma-i+1$ bits of $y$ can be completed by traversing down from node $\alpha$. 

For that we  continue by querying the bitvector at node $\alpha$ with the same interval which we already used for querying node $\alpha$, but this time using bit value $0$ instead of $1$. 

We then continue to traverse down the tree, for each 
traversed node querying the bitvector for bit value $1$. If the returned interval is non-empty, we continue 
traversing down the tree with the interval. Otherwise we query the for bit value $0$ and continue traversing with 
the returned interval (which is necessarily non-empty). When we reach the leaf, we will have constructed 
the whole value of coordinate $y$ and will have a singleton interval $[1,1]$. In order to reconstruct the value of $x$, we climb up the tree retraversing (in reverse order) the nodes we already traversed, and for a node at level $i$ issue a $\select$ query using the bit value $y[i]$ for the single position that we obtained from the previous $\select$ query (or position $1$ for the first $\select$ query). 
At this point we will have determined both coordinates $x$ and $y$ of the answer.

\section{Range-predecessor query answering}
\label{sec:range_succ_queries}
We now describe how we answer range predecessor queries with the data structure
used in Theorem~\ref{range_pred_theo1}. The algorithm can be thought of as a generalization 
of the query algorithm used for the wavelet tree. 
We are given integers $x_1,x_2$ and $y_2$ and must find the point $(x,y)$
such that $x_1\leq x\leq x_2$, $y\leq y_2$ and $y$ is maximal. 
A query will first proceed by traversing the tree top-down, 
where at level $i$ a $\rank$ query at a node $\alpha_i$ will allow 
to determine the range in $Y_{\alpha_{i+1}}$ 
from the range in $Y_{\alpha_i}$, where $\alpha_{i+1}$ is the next node at level $i+1$. 
The range minimum query data structure at level $i$ will allow to determine 
whether the $y$ coordinate of the answer shares at least $i$ chunks with $y_2$. 
Once it has been determined that the $y$ coordinate shares a chunk of length 
$t$ with $y_2$, then the value of the next chunk (chunk number $i+1$) 
of $y$ should be smaller than the corresponding chunk of $y_2$ and then the next chunks 
will all need to have maximal values. Hence, we will use range maximum 
queries for at all next levels. With the $y$ coordinate determined
we can read the $x$ coordinate from $X[y]$. 

We now give a detailed description of how the queries are solved. 
Before delving into the details, we first show how to handle the easy case 
in which the answer is a point $(x,y)$ such that $y=y_2$. To eliminate 
the case, it suffices to test that $x_1\leq X[y_2]\leq x_2$, and if it is,
then the answer is $(X[y_2],y_2)$. We now show how 
queries are answered under the assumption that $X[y_2]\notin [x_1,x_2]$. 

We traverse the tree top-down, successively for the root node, 
then the node $\alpha_1=y_2[1..\sqrt{\log n}]$, then the node $\alpha_2=y_2[1..2\sqrt{\log n}]$
and so on. For the root node, we first compute the value $m_0=\mathtt{min}(Y[x_1..x_2])$. 
Then query the predecessor data structure $P_{(0,\alpha_1)}$ to check whether 
$Y[x_1,x_2]$ contains the value $\alpha_1$. The predecessor data structure will then be able to return 
a pair of values $(x_{1,0},x_{2,0})$ where $(x_{1,0})$ (resp. $(x_{2,0})$) 
is the leftmost (resp. rightmost) position in $[x_1,x_2]$ such that $Y[x_{1,0}]=\alpha_1$ 
(resp. $Y[x_{2,0}]=\alpha_1$). If the interval $Y[x_1,x_2]$ does not contain the value $\alpha_1$, 
we stop at the first level. Otherwise, we go to the second level to node $\alpha_2$, compute $m_1=\mathtt{min}(Y[x_{1,0}..x_{2,0}])$, and query the predecessor 
data structure $P_{(\alpha_2,y_2[1..2\sqrt{\log n}])}$ for the pair 
$(x_{1,0},x_{2,0})$ to check whether $Y_{\alpha_1}[x_{1,0},x_{2,0}]$
contains the value $y_2[\sqrt{\log n}+1..2\sqrt{\log n}]$. If that is the case, then
the predecessor data structure will return a pair $(x_{1,1},x_{2,1})$ such 
that $(x_{1,1})$ (resp. $(x_{2,1})$) is the leftmost (resp. rightmost) position 
in $[x_{1,0},x_{2,0}]$ with $Y[x_{1,1}]=y_2[\sqrt{\log n}+1..2\sqrt{\log n}]$ 
(resp. $Y[x_{2,1}]=y_2[\sqrt{\log n}+1..2\sqrt{\log n}]$). 
We continue the traversal of the tree in the same way until the  
predecessor query fails at a certain level $i$\footnote{Note that 
one of the predecessor queries will have to fail, because we have 
eliminated the case that the $y$ component of the query answer 
equals $y_2$.} 

We let $j$ be the deepest level such that $m_j<y_2[(j-1)\sqrt{\log n}+1..j\sqrt{\log n}]$. This tells us 
that the final value for $y$ is prefixed by $y_2[1..(j-1)\sqrt{\log n}]$, 
but is followed by a chunk that differs from  (more precisely, is strictly smaller than) 
$y_2[(j-1)\sqrt{\log n}+1..j\sqrt{\log n}]$. 

Then we query the wavelet tree $W_\alpha$ with
$\alpha=\alpha_j$ to find the range predecessor of $y_2[(j-1)\sqrt{\log n}+1..j\sqrt{\log n}]$
in the interval $Y_\alpha[x_{1,j}..x_{2,j}]$. This will produce the next $\sqrt{\log n}$ bits
of $y$ (denoted by $y'$) and an interval $(x_{1,j+1},x_{2,j+1})$. We then continue to the node 
$\alpha=y_2[1..(j-1)\sqrt{\log n}]\cdot y'$, but this time the next bits of $y$ will be produced 
by the range maximum query data structure over interval $Y_\alpha[x_{1,j+1}..x_{2,j+1}]$. 
We continue traversing the tree in this way until the bottom node,  
at which point we will have induced the full value of $y$. To get $x$ we simply read $X[y]$. 
 
\subsection{Faster queries}
\label{subsec:faster_range_pred}
We now show how queries are answered using the data structure
from Theorem~\ref{range_pred_theo2}. The algorithm can be thought of as generalization 
of the query algorithm presented in beginning of Section~\ref{sec:range_succ_queries}. 
Recall that we are given integers $x_1,x_2$ and $y_2$ and must find the point $(x,y)$
such that $x_1\leq x\leq x_2$, $y\leq y_2$ and $y$ is maximal. 
As usual we eliminate the trivial case that $x_1\leq X[y_2]\leq x_2$, 
in which case the answer is $(X[y_2],y_2)$. 
As before the query proceeds in two phases. The first phase
determines the longest common prefix between $y$ and $y_2$
and the second phase allows us to determine the remaining bits 
of $y$. Finally, the value of $x$ can be determined by reading $X[y]$. 
We now give details of the two phases. 
In the first phase, the longest common prefix between $y$ and $y_2$, is determined 
in chunks of $\log^{(c-1)/c}n$ bits by traversing the tree of granularity 
$\log^{(c-1)/c}n$. Then at most $\log^{(c-1)/c}n-1$ bits will remain to be determined, 
and we continue from a node in the tree at level of granularity $\log^{(c-2)/c}n$, 
labelled by $y_2[(i-1)\log^{(c-1)/c}n+1,i\log^{(c-1)/c}n]$, for some integer $i$.
Since we have to determine less than $\log^{(c-1)/c}n$ bits, the number of traversed 
tree levels will be less than $\log^{1/c}n$. We then continue refining the length of 
longest common prefix by traversing trees at decreasing levels of granularity until 
we reach a node at level of granularity $1$, which is an ordinary wavelet tree. 
As before at each node, we will use a range-minimum query to determine whether the label 
of the node is a prefix of the longest common prefix of $y$ and $y_2$ and use the  
predecessor query to determine whether we continue exploring the next node at the next tree level
and the interval to be used at that next node. 
In the second phase, we will traverse trees of increasing level of granularities, 
from the node at which the first phase has stopped. 
This time we will use range-maximum queries to determine both the next node 
to explore and the next chunk of bits of $y$ (the chunk length being the 
level of granularity). We switch from a level of granularity $\log^{d/c}n$ to the next 
one of granularity $\log^{(d+1)/c}n$, whenever the number of determined bits of $y$ 
is multiple of $\log^{(d+1)/c}n$ and the first node at that level will be the one 
labelled by the bits of $y$ which have been determined so-far. 
As before the range to be used at next tree level will be determined using predecessor 
query on the range at current node, using the current chunk determined from the 
range-maximum query. 
It is easy to see that the query time of both phases is $O(c\log^{1/c}n\log\log n)$, 
since we have $c$ levels of granularity and at each such level we traverse 
at most $\log^{1/c}n$ nodes, spending $O(\log\log n)$ time at each node. 
This finishes the description and analysis of the queries. 

We traverse the tree top-down, successively for the root node, 
then the node $\alpha_1=y_2[1..\sqrt{\log n}]$, then the node $\alpha_2=y_2[1..2\sqrt{\log n}]$
and so on. For root node, we first compute the value $m_0=\mathtt{min}(Y[x_1..x_2])$. 
Then query the predecessor data structure $P_{(0,\alpha_1)}$ to check whether 
$Y[x_1,x_2]$ contains the value $\alpha_1$. The predecessor data structure will then be able to return 
a pair of values $(x_{1,0},x_{2,0})$ such that $(x_{1,0})$ (resp. $(x_{2,0})$) 
is the leftmost (resp. rightmost) position in $[x_1,x_2]$ such that $Y[x_{1,0}]=\alpha_1$ 
(resp. $Y[x_{2,0}]=\alpha_1$). If the interval $Y[x_1,x_2]$ does not contain the value $\alpha_1$, 
we stop at the first level. Otherwise, we go to the second level to node $\alpha_2$, compute $m_1=\mathtt{min}(Y[x_{1,0}..x_{2,0}])$ and query the predecessor 
data structure $P_{(\alpha_2,y_2[1..2\sqrt{\log n}])}$ for the pair 
$(x_{1,0},x_{2,0})$ to check whether $Y_{\alpha_1}[x_{1,0},x_{2,0}]$
contains the value $y_2[\sqrt{\log n}+1..2\sqrt{\log n}]$. If that is the case, then
the predecessor data structure will return a pair $(x_{1,1},x_{2,1})$ such that 
$(x_{1,1})$ (resp. $(x_{2,1})$) is the leftmost (resp. rightmost) position 
in $[x_{1,0},x_{2,0}]$ such that $Y[x_{1,1}]=y_2[\sqrt{\log n}+1..2\sqrt{\log n}]$ 
(resp. $Y[x_{2,1}]=y_2[\sqrt{\log n}+1..2\sqrt{\log n}]$). We continue the traversal 
of the tree in the same way until the predecessor query fails at a certain level 
$i$\footnote{note that one of the predecessor queries will have to fail, because 
we have eliminated the case that the $y$ component of the query answer equals $y_2$. 
To see why notice that the last query is for node $y_\alpha$ 
where $\alpha=y_2[1..\log n-\sqrt{\log n}]$ and the last query 
is for $y'=y_2[\log n-\sqrt{\log n}+1..\log n]$}.

\section{Elias-Fano based predecessor data structure}
\label{sec:Elias_Fano_pred}
We now show how Elias-Fano predecessor data structures are built. 
Suppose that we have a set $S$ of $n$ keys from interval $[1..u]$
(where for simplicity $u$ and $n$ are powers of two). We can 
show a data structure that uses $n(2+\log(u/n))+o(n)$ bits of space 
and that allows us to answer to predecessor ($\rank$) queries
in $O(\log \log u)$ time as well as finding the key of rank $i$ ($\select$ queries) 
in constant time. The Elias-Fano encoding is composed of 
two substructure. Let $x_1<x_2<\ldots <x_n$ be the sequence 
of keys to be encoded with $x_i\in [1..u]$ for all $i\in[1..n]$. The first substructure is an array 
$A[1..n]$, where $A[i]$ contains the least significant 
$\log(u/n)$ bits of $x_i$. The second substructure is a bitvector
V of length $2n$ bits which contains the sequence $0^{x'_1}10{x'_2-x'_i}\ldots 0^{x'_n-x'_{n-1}}1$, 
where $x'_i=x_i/(u/n)$. In other words, the bitvector $V$ encodes the most significant $\log n$ bits 
of elements $x_i$. In order to support the $\select$ operation for position $i$
(computing $x_i$), we can first go to $A[i]$ to retrieve the least significant $\log(u/n)$
bits of $x_i$ and then do a $\select$ query to find the location of the $i$th one 
in $V$. If that location is $j$ then the value of the most significant $\log n$ bits 
of $x$ are equal to the number of zeros before position $j$, which is 
$j-i$. Thus a $\select$ query can be answered in constant time. 
To answer to $\rank$ queries, we will build a $y$-fast trie 
predecessor data structure~\cite{Wi83} on the set $S'$ of $n/\log^2 u$ 
keys $x_1,x_{\log^2u+1},\ldots,x_{n}$. 
This data structure will occupy $O(n/\log u)=o(n)$ bits of space and allows 
to determine the prededecessor of a query key $x$ in $S'$ in $O(\log\log u)$ time. 
This will allow to restrict the predecessor search in $S$ to a small interval 
of size $n$. The predecessor search can then be completed by a binary search 
over that set, by doing $\select$ queries. This binary search also takes 
$O(\log\log u)$ time. 

The data structure can be generalized as follows. 
Given a parameter $v\leq u$ (again assume that $v$ is a power of two), 
we can build a data structure that occupies $v+n(1+\log u/v)+o(v+n)$ bits 
of space and that answers to $\rank$ and $\select$ 
queries within the same amount of time. To implement the data structure, 
we will use a bitvector $V$ of size $n+v$ bits with $n$ ones 
and $v$ zeros (the ones and zeros are stored as was defined before
except that $x'_i=x_i/(u/v)$) and store in $V$ the 
least significant $(u/v)$ bits of each key. The query time bounds
are preserved. 
\begin{lemma}
Given a set $S\subset [1..u]$ with $|S|=n$ and a number $v\leq u$, 
we can build a data structure which occupies
$v+n(1+\log u/v)+o(v+n)$ (assuming $n$,$v$ and $u$ are powers of 
two) and answers to $\rank$ queries in time $O(\log\log u)$ 
and $\select$ queries in constant time. 
\end{lemma}

\subsection{Simple Construction}
~\label{sec:simple_Elias_Fano_build}
We can now show the following lemma: 
\begin{lemma}
\label{lemma:simple_Elias_Fano_build}
Given a sequence $Y$ of length $n$ over alphabet $[1..\sigma]$ (where $\sigma\leq n$ and 
both $\sigma$ and $n$ are powers of two), we can build $\sigma$ predecessor data structures 
so that the data structure number $c$ stores the positions of character $c$ in the sequence $Y$ such that: 
\begin{enumerate} 
\item The total space occupied by all predecessor data structures is $n(2+\log\sigma)+o(n)$ bits of space. 
\item The total construction time of the data structures is $O(n)$ time. 
\item A predecessor ($\rank$) query is answered in time $O(\log\log n)$ and 
a $\select$ query is answered in constant time
\end{enumerate}
\end{lemma}
\begin{proof}
We will build $\sigma$ Elias-Fano data structures (generalized as above)
denoted $E_\alpha$ for $\alpha\in [1..\sigma]$. 
For each data structure we set $v=n/\sigma$. The space used by data structure 
$\alpha$ is $v+n_\alpha(1+\log(n/v))+o(v+n_\alpha)=v+n_\alpha(1+\log\sigma)+o(v+n_\alpha)$. 
Since $n_\alpha$ and $v$ sum up to $n$, we get that the total space usage is $n(2+\log\sigma)+o(n)$
bits of space. The vectors $V_\alpha$ and $A_\alpha$ can be built easily in $O(n_\alpha)$. 
For that, we can first build the sequence of positions $P_\alpha$ from $Y$ (initially the sequences 
are empty), by scanning $Y$ and appending for each $Y[i]=c$ append position $i$ to sequence $P_\alpha$. 
Then, building $V_\alpha$ and $A_\alpha$ is done by scanning $P_\alpha$ and for each element 
writing its least significant $\log\sigma$ bits in $A_\alpha$ and writing a number of zeros 
and a one in $V_\alpha$
(the number of zeros is based on the $\log n-\log\sigma$ most significant bits). 

In order to support $\rank$ queries, for each $E_\alpha$ containing at least $\log^2 n$ elements, we sample every $\log^2 n$th occurrences of character $\alpha$ in $Y$ storing the resulting positions in a $y$-fast trie $\mathtt{Tr}_\alpha$. A predecessor query on $E_\alpha$ is then solved by first querying $\mathtt{Tr}_\alpha$, which will answer in time $O(\log \log n)$ and complete with binary search for an area of length at most $\log^2 n$ doing $O(\log\log n)$ $\select$ queries on $E_\alpha$. The construction of $\mathtt{Tr}_\alpha$ clearly takes $O(1+|E_\alpha|)$ and the space is clearly $O(\lfloor|E_\alpha|/\log n \rfloor)$ bits of space. When added over all $\alpha\in[1..\sigma]$, the construction time for $\rank$ and $\select$ structures on $V_\alpha$ sums up to $O(n)$ and for all $\mathtt{Tr}_\alpha$ sums up to $O(\sigma+n/\log^2 n)$. This finishes the proof of the lemma.  
\qed
\end{proof}

\subsection{Bit-parallel Construction}
~\label{sec:bit_parallel_Elias_Fano_construction}
We now exploit the bitparallelism to show a faster construction, showing the following lemma: 
\begin{lemma}
\label{lemma:bit_parallel_elias_fano_build}
Given a sequence $Y$ of length $n$ over alphabet $[1..\sigma]$ 
(where $\sigma$ and $n$ are both powers of two) and a global precomputed table of size $N$ 
such that $\sigma\leq n\leq N\leq 2^w$, we can build $\sigma$ predecessor data structures 
so that the data structure number $c$ stores the positions of character $c$  in the sequence $Y$ such that: 
\begin{enumerate} 
\item The total space occupied by all predecessor data structures is $n(2+\log\sigma)+o(n)$ bits of space. 
\item The total construction time of the data structures is $O(\frac{n(\log\sigma)^2}{\log N}+\sigma)$. 
\item The temporary space used during the construction is $O(n\log\sigma)$ bits. 
\item A predecessor ($\rank$) query is answered in time $O(\log\log N)$ and 
a $\select$ query is answered in constant time
\end{enumerate}
The global precomputed table can be shared by many instances of the data structures. 
\end{lemma}
\begin{proof}
We will build $\sigma$ Elias-Fano data structures (generalized as above)
denoted $E_\alpha$ for $\alpha\in [1..\sigma]$. 
For each data structure we set $v=n/\sigma$. The space used by data structure 
$\alpha$ is $v+n_\alpha(1+\log(n/v))+o(v+n_\alpha)=v+n_\alpha(1+\log\sigma)+o(v+n_\alpha)$. 
Since $n_\alpha$ and $v$ sum up to $n$, we get that the total space usage is $n(2+\log\sigma)+o(n)$
bits of space. 
As in Lemma~\ref{lemma:simple_Elias_Fano_build}, we will build $\sigma$ Elias-Fano data structures (generalized as above)
denoted $E_\alpha$ for $\alpha\in [1..\sigma]$, in which we set we set $v=n/\sigma$ for each data structure. 
As shown above the total space usage of the data structures sums up to $n(2+\log\sigma)+o(n)$. 

We now describe the construction procedure. The construction proceeds in $\log\sigma$
phases. We describe the first phase, the subsequent phases will be (almost) identical, but 
will have different input and output. 
Before doing the first phase, we construct a vector $A[1..n]$
such that $A[i]=i\bmod \log\sigma$ and a bitvector $V=(01^{\log\sigma})^{n/\log\sigma}$. Notice that these two vectors represent together the Elias-Fano representation of the sequence $1,2,\ldots n$. 
The phase will have as input the arrays $A$, $V$ and the sequence $Y$ and will output a pair of vectors $(A_0,A_1)$, 
a pair of bitvectors $(V_0,V_1)$, and a pair of sequences $(Y_0,Y_1)$. The sequence $Y_0$ (resp. $Y_1$) will store the subsequence 
of characters from $Y$ which belong to the first (resp. second) half of the alphabet. The array $A_0$ (resp. $A_1$) will store the values from $A$ whose corresponding positions in $Y$ belong to first (resp. second) half of the alphabet. Finally the bits of $V$ are copied into $V_0$ and $V_1$. More precisely every $0$ in $V$ is copied to both $V_0$ both $V_1$ while every $1$ is copied to either $V_0$ or $V_1$. The vector $V$ is scanned right-to-left and the $i$th $1$ in $V$ is in correspondence with the $i$th element of $Y$. If $Y[i]$ belongs to the first (resp. second) half of the alphabet (this can be checked by looking at most significant bit of $Y[i]$), then it is appended to $V_0$ (resp. $V_1$). Whenever a $0$ is encountered in $V$, it is appended to both $V_0$ and $V_1$. One can now easily see that $A_0$ and $V_0$ (resp. $A_1$ and $V_1$) is the Elias-Fano representation of the occurrences of characters from the first (resp. second) half of the alphabet in the sequence $Y$. The first phase can thus easily construct the output by simultaneously scanning $A$, $Y$ and $V$ in left-to-right order and appending at the end of $A_0$, $ A_1$, $V_0$, $V_1$, $Y_0$ and $Y_1$. In order to efficiently implement the first phase, we will make use of the four Russian technique. We read $A$, $V$ and $Y$ into consecutive blocks of $b=\lceil\log N/3\log\sigma\rceil$ consecutive elements, which occupy $b(1+2\log\sigma)$ bits of space. We use a lookup table which for each combination of $3$ blocks from $A$, $V$ and $Y$, will indicate the blocks to be appended at the end of $A_0$, $A_1$, $V_0$, $V_1$, $Y_0$ as well as by how much we advance the scanning pointers into $A$ and $Y$ (which will be the number of $1$s in the block read from $V$). The information stored for every combination easily fits into $O(\log n)=O(\log N)$ bits of space and the total number of possible combinations 
is $O(2^{b(1+2\log\sigma)})=O(2^{2\log N/3+\log_\sigma N})$ and thus the lookup table occupies $O(2^{2\log N/3+\log_\sigma N}\log N)$ bits of space. In the second phase, we build a pair of vectors $(A_{00},A_{01})$ (resp. $(A_{10},A_{11})$), 
a pair of bitvectors $(V_{00},V_{01})$ (resp. $(V_{10},V_{11})$), and a pair of sequences $(Y_{00},Y_{01})$ (resp. $(Y_{10},Y_{11})$) based on $A_0$, $V_0$ and $Y_0$ (resp. $A_1$, $V_1$ and $Y_1$). The procedure is similar to the one  used in the first phase, except that now the distribution of elements in the output is based on the second most significant bit of elements of $Y_0$ (resp. $Y_1$). The lookup table for the second occupies the same space as the one built in the first phase. At the end of the second phase the pairs $(A_{00},V_{00})$, $(A_{00},V_{00})$, $(A_{01},V_{01})$, $(A_{10},V_{11})$  will represent Elias-Fano representations of the occurrences of characters from respective subalphabets $[1..\sigma/4]$, $[\sigma/4+1..\sigma/2]$, 
$[\sigma/2+1..3\sigma/4]$ and $[3\sigma/4+1..\sigma]$ in the original sequence $Y$. We continue in the same way doing $\log\sigma-2$ more phase, where at phase $i$ we will have built Elias-Fano for sub-alphabets of size $\sigma/2^i$. At the end we will get the Elias-Fano representation of the occurrences of each distinct character in the sequence $Y$. 
We now analyze the time and space used by the algorithm. 
The total space used by all the $O(\log\sigma)$ lookup tables will be $O(2^{2\log N/3+\log_\sigma N}\log N\log\sigma)=o(N)$ bits of space. For the running time it is easy to see that a phase $i$ runs in time $O(2^i+n/b)$ time, since we have $2^i$ sub-alphabets and we process $b$ elements of each processed vector in constant time. Over all $\log\sigma$ phases the total running 
time sums up to $O(\sigma+n\log\sigma/b)=O(\sigma+\frac{n(\log\sigma)^2}{\log n})$. 
Let $E_\alpha=(V_\alpha,A_\alpha)$ be the Elias-Fano representation for occurrences of character $\alpha\in[1..\sigma]$. 
In order to complete the construction, we need to construct the support for $\rank$ and $\select$ 
queries on each vector $V_\alpha$. This can be done in time $O(1+|V_\alpha|/\log N)$ for the bitvector $|V_\alpha|$, by using the construction in~\cite{BGKS15}. This allows us to support $\select$ queries on $E_\alpha$ in constant time. 
In order to support $\rank$ queries, for each $E_\alpha$ containing at leat $\log^2 N$ elements, we sample every $\log^2 N$th occurrences of character $\alpha$ in $Y$ by doing $\select$ queries on $E_\alpha$ for positions $1,\log^2 N,2\log^2 N,\ldots$ and store the resulting positions in a $y$-fast trie $\mathtt{Tr}_\alpha$. A predecessor query on $E_\alpha$ is then solved by first querying $\mathtt{Tr}_\alpha$ which will answer in time $O(\log \log n)$ and complete with binary search for an area of length at most $\log^2 N$ doing $O(\log\log N)$ $\select$ queries on $E_\alpha$. The construction of $\mathtt{Tr}_\alpha$ clearly takes $O(1+|E_\alpha|/\log N)$ and the space is clearly $O(\lfloor|E_\alpha|/\log N \rfloor)$ bits of space. When added over all $\alpha\in[1..\sigma]$, the construction time for $\rank$ and $\select$ structures on $V_\alpha$ sums up to $O(\sigma+n/\log N)$ and for all $\mathtt{Tr}_\alpha$ sums up to $O(\sigma+n/\log^2 N)$. This finishes the proof of the lemma.  
\qed
\end{proof} 

\section{Sampled range-minimum queries}
\label{sec:sampled_rmq}
We first start by constructing a sampled sequence $S'$ from the sequence $S$ we want to index. 
Initially $S'$ is empty. Suppose that we allow $o(N)$ precomputation of a table of size $o(N)$ bits 
that can be shared by all instances of the range minimum or maximum queries. We assume that $N\geq\sigma$. 
We divide $S$ into blocks of size $b=\frac{\log N}{2\log\sigma}$ elements each, then scan it in left-to-right order, compute minimum element in each block and append the result at the end of $S'$. We then build 
a rmq $R'$ on top of $S'$ in time $O(\lceil\frac{n\log\sigma}{\log N}\rceil )$. To compute the minimum in each block, 
we use a lookup table that returns the minimum element on all possible blocks of size $b$. The table stores $O(\sigma^b)$ 
elements each of length $\log\sigma$ bits, for a total space of $O(\sigma^{b}\log\sigma)=o(N)$ bits. 

To answer a rmq query for $S[i..j]$, we first check whether the query is contained in one or two blocks of $S$. If that is the case, then we read the blocks  compute the minimum in each block in constant time, and compute the minimum of all the $2$ block also in constant time. For that we use a precomputed table that will contain $O(\sigma^{b}\log^2 b)$ elements each using $\log b$ bits. The table contains the answers to all possible queries ($O(\log^2b)$) over all possible blocks ($\sigma^b$). 
The space is $O(\sigma^{b}\log^2b\log\sigma)=o(N)$ bits. 
If the query spans more than two blocks, then we can divide it into three parts. A head and tails parts which are contained in one block each and a central part, whose length is multiple of a block length. We answer compute the minimum on head and tail 
in constant time, using the precomputed table and compute the minimum on the central part using $R'$, and finally take the minimum of the three. 

We thus have proved the following lemma: 
\begin{lemma}
\label{fast_rmq_build_lemma}
Given a sequence $Y$ of length $n$ over alphabet $[1..\sigma]$ (where $\sigma\leq n$ and both $\log\sigma$ and $n$ 
are powers of two) and global precomputed tables total of size $o(N)$, where $N\leq 2^w$, we can build a range min data structure on top of $Y$: 
\begin{enumerate} 
\item The space used by the data structure is $O(n/\log_\sigma N)$ bits. 
\item A query is answered in constant time. 
\item The data structure can be constructed in time $O(n/\log_\sigma N)$. 
\item The temporary space used during the construction is $O(n/\log_\sigma N)$ bits. 
\end{enumerate}
The precomputed tables can be shared by many instances of the data structure and can be built in $o(N)$ time. 
\end{lemma}
\section{Construction of the fast range-predecessor data structure}
In this section, we show the construction of the data structure used in Theorem~\ref{range_pred_theo2}. 
The main ingredients are lemmas~\ref{fast_rmq_build_lemma} and~\ref{lemma:bit_parallel_elias_fano_build}. 
We let $L_1=\log^{i/c}n$ be the level of granularity  such that $\log^{i/c}n\geq\sqrt{\log N}$ and 
$L_2=\log^{(i-1)/c}n<\sqrt{\log N}$. For level $L_1$ and higher, we will use simple linear time 
algorithms to build the predecessor and range minimum (maximum) data structures. The time for these levels 
is dominated by the time used for level $L_1$ for which the total number of elements stored in all sequences (and hence in predecessor and range minimum or maximum data structures) at all nodes will be $O(n\log n/L_1)=O(n\frac{\log n}{\sqrt{\log N}})$. The sum of alphabet sizes at all nodes is clearly $O(n)$. Hence the construction time at that level will be $O(n(1+\frac{\log n}{\sqrt{\log N}}))$. For the next higher level the construction time will be $O(n(1+\frac{\log n}{\log^{1/c}n\sqrt{\log N}}))$. Continuing the same way we deduce that the construction time for these levels will be $O(n(c_1+\frac{\log n}{\sqrt{\log N}}))$, 
where $c_1$ is the number of levels. 

For level $L_2$ and lower we will use bitparallel construction algorithm. For an given node $\alpha$ containing 
$n_\alpha$ elements, the construction time will be $O(\sigma+n_\alpha\frac{(\log\sigma)^2}{\log N})$ (which is dominated by the time for construction the Elias-Fano predecessor data structure), where $\sigma=2^{L_2}=O(2^{\sqrt{\log N}})$. The total sum of the $\sigma$ term over all nodes is $t_1=O(n)$ and the total sum of the terms 
$n_\alpha\frac{(\log\sigma)^2}{\log N}$ will be $t_2=O(\frac{n}{L_2}\cdot\frac{(\log\sigma)^2}{\log N})$, since the total
number of elements stored in all nodes is $O(n/L_2)$. Thus, we have $t_2=O(\frac{n\log nL_2}{\log N})$. Since 
we have $L_2\leq \sqrt{\log N}$, we conclude that $t_2=O(\frac{n\log n}{\sqrt{\log N}})$. We notice that $t_2$ will decrease 
by factor $\log^{1/c}n$ each time we go to the next smaller level of granularity. The term $t_1$ remains stable. Since, the term $t_2$ of level $L_2$ will dominate the terms $t_2$ of all lower levels, and the term $t_1=O(n)$ remains stable over the lower levels, we conclude that the total construction 
time for level $L_2$ and lower levels will be $O(n(c_2+\frac{\log n}{\sqrt{\log N}}))$, where $c_2$ is the number of levels. 
Summing up the construction times of all levels we get the construction time stated in Theorem~\ref{range_pred_theo2}.

\section{Rightmost Parsing}

\subsection{The full picture}
\label{sec:time_efficient_rightmost}
We now present the complete algorithm. 
We first categorize each query.
Queries that fall in the same block are redirected 
to the range-predecessor data structure responsible 
for handling that block and queries that are longer 
than $\ell$ are put aside to be solved by the data structure 
for handling long factors. 
For the remaining factors (short factors with ranges 
that cross block boundaries) we will construct the tree of queries. 
This can be done 
in $O(n)$ time, by sorting the query ranges 
first by starting positions and then by the negation
of their ending positions. The sorting can be done 
in $O(\sqrt{n}+n/\log_\sigma n)$ time 
using 
radix-sort. 

The rest of the algorithm is straightforward, except for 
few details on how the points 
stored in the range-predecessor data structure are constructed. 
This is done as follows. The generated 
points are pairs consisting of a suffix array position ($x$ coordinate)
and text position ($y$ coordinate). 
We first notice that 
we need to divide the points according
to their $x$ coordinates, such that the points
whose $x$ coordinate is in $[1..B]$ go to the first range-predecessor 
data structure, points with coordinate in $[B+1,2B]$ 
go to second data structure and so on. 
The points are generated by inverting the Burrows-Wheeler 
transform in left-to-right order. This generates 
the values of $y$ coordinates that appear in each block in sorted order. 
Finally, for each block, we need as input to building its range-predecessor 
data structure, the 
points sorted by $x$ coordinate and for each point the $y$ coordinate 
replaced by the rank of the value among all values of $y$ that appear 
inside the block. Since we extract the points by increasing value 
of $y$, we can keep an array $C[1..n/B]$ that stores the number of points 
that have been extracted so far for each block. The counter is incremented 
each time we extract a point, and the rank of the $y$ value of a point 
that has been just extracted in block $i$ will equal the counter $C[i]$. 
Finally we need to sort all the extracted points by their $x$ coordinates, 
and this done via radix-sort. 
The algorithm uses $O(n\log n)$ bits and $O(n(1 + \log\sigma/\sqrt{\log n}))$ time 
overall~\footnote{Notice the time is deterministic and not randomized.  
The source of randomization is the construction and indexation of the BWT, 
which is subsequently inverted. However this is not needed anymore, since we can 
build the suffix array of $\X$ in deterministic linear time and $O(n\log n)$ bits of space.}.

\subsection{Optimal space}
\label{sec:opt_space_rightmost}
It remains to show how to reduce the space to the optimal $O(n\log\sigma)$ bits. 
The additional space used for long factors (Section~\ref{sec:rightmost_long_factors}) 
and short factors crossing block boundaries (Section~\ref{sec:rightmost_sparse_tree}) 
is $O(n)$ bits. The bottleneck is for the factors with small range (
see Section~\ref{sec:rightmost_rem_factors}) in which the used space is 
$O(n\log n)$. In what follows, we show how to reduce the space to $O(n\log\sigma)$
bits. 

To this end we divide the range of $y$ into $\log n/2$ equal sub-ranges. 
We now build the tree of all queries that do not cross 
a block boundary in the same way we did for the short factors that cross a block boundary
in Section~\ref{sec:rightmost_sparse_tree}. 
As before, for every leaf of the suffix tree, 
we associate the nearest ancestor that corresponds to a query range. Additionally for each node 
in the query tree, we associate a bitmap $B_\alpha[1..\log n/2]$ of size $\log n/2$ bits. The 
query tree occupies $O(n\log\sigma)$ bits. We now invert the BWT and, for each leaf that has 
been extracted, we mark the bit $B_\alpha[i]$, where $i$ is the block in which the text position 
falls and $\alpha$ is the internal node pointed to by the leaf. We finally do a complete postorder 
traversal of the tree so that the bitmap of each node is bitwise OR'd with all the bitmaps of its 
children. At the end, for every query range $[l_\alpha,r_\alpha]$ corresponding to a node $\alpha$, 
$B_\alpha$ will mark every subrange of $y$ where there exists a point $(a,b)$ such that 
$a\in [l_\alpha,r_\alpha]$ and $b$ belongs to the subrange of $y$. 

We now build $\log n/2$ query trees (henceforth local query trees). For every node of the tree
(henceforth main query tree), 
we know that all range-predecessor queries will have the same interval 
$[l_\alpha,r_\alpha]$ on the $x$ axis, but a different $b$ on the $y$ axis. 

For every query, the answer can only be in two sub-ranges: the sub-range 
$R_b$ that contains $b$ if $B_\alpha[R_b]$ is marked  
or the largest $R'_b<R_b$ such $B_\alpha[R'_b]$ contains a marked bit. 
The $\log n/2$ local query trees are built via a preorder traversal 
of the main query tree and for every query determining the one or two 
target local query trees to which the query should be copied. 
That is a query will be copied to the local query trees $R_b$ 
(if $B_\alpha[R_b]$ is marked) and $R'_b$ (If $R'_b$ exists). 

The nodes of the query trees (and the queries attached to them) 
are built on-the-fly during the traversal 
of the main query tree, by keeping a stack for every query tree. 
For every query in a local query tree, we keep a pointer to the originating 
query in the main query tree.
Note that the total space used by all the local query trees remains bounded 
by $O(n\log\sigma)$ bits, which is the same as the main query tree. 

We will now apply the algorithm described in Section~\ref{sec:rightmost_rem_factors} 
on every subrange of $y$. In more detail, we invert the BWT left-to-right. 
We do this in $\log n/2$ phases, where in phase $i$ we will 
compute all the positions in $\SA$ that point to text positions in 
$[(i-1)n/C+1..iC]$ with $C=2n/\log n$. That is, we output all pairs $(x,\SA[x])$, 
such that $\SA[x] \in [(i-1)n/C+1..iC]$. During the course of phase $i$, we apply 
the algorithm of Section~\ref{sec:rightmost_rem_factors} verbatim, 
except that now the points are only the one with $y\in[(i-1)n/C+1..iC]$ 
and the queries are from the $i$th local query tree and not 
the main query tree. 

It is easy to see that the time complexity stays exactly as it was before. 
The construction time of the range-predecessor and predecessor data structures 
for every phase will now be $O((n/\log n)(\log\sigma/\sqrt{\log n}))$. 
In particular the radix sort done on the $x$ values can now be executed 
in $O(\sqrt{n}+n/\log n)$ time and using $O(\sqrt{n})$ words of space. 
The space usage of the range-predecessor and predecessor structures over all 
phases will be just $O(n)$ bits. Both structures are destroyed at the end of 
each phase, so that the total peak space usage of the whole algorithm is $O(n\log\sigma)$ bits.

\end{document}